\documentclass[conference]{IEEEtran}
\ifCLASSINFOpdf
\else
\fi
\DeclareMathAlphabet{\mathpzc}{OT1}{pzc}{m}{it}  
\usepackage{mathrsfs}
\usepackage{amsthm}
\usepackage{makeidx}
\usepackage{examples}
\usepackage[pdftex]{graphicx}  
\usepackage{amssymb}
\usepackage{graphicx}
\usepackage{algorithm}
\usepackage{epstopdf}
\usepackage{color}
\usepackage{url}
\graphicspath{{../eps/}{../ps/}}
\usepackage{psfrag}    
\usepackage{algorithm}
\usepackage{algorithmic}
\usepackage{amsmath}
\usepackage{hyperref}

\usepackage[numbers,sort&compress,square]{natbib}
\usepackage[utf8]{inputenc}
\usepackage{tikz}

\font\msbm=msbm10 at 10pt
\newcommand{\ZZ}{\mbox{\msbm Z}}

\newcommand{\NN}{\mbox{\msbm N}}

\newcommand{\FF}{\mbox{\msbm F}}
\def \Z {{\ZZ}}

\def \N {{\NN}}
\def \F {{\FF}}

\newtheorem{theorem}{Theorem}
\newtheorem{lemma}[theorem]{Lemma}
\newtheorem{remark}[theorem]{Remark}
\newtheorem{example}[theorem]{Example}
\newtheorem{definition}[theorem]{Definition}
\newtheorem{problem}[theorem]{Problem}

\begin{document}
%
\title{Tradeoff for Heterogeneous Distributed Storage Systems between Storage and Repair Cost}
\author{\IEEEauthorblockN{Krishna Gopal Benerjee and Manish K. Gupta}
\IEEEauthorblockA{Laboratory of Natural Information Processing,\\
Dhirubhai Ambani Institute of Information and Communication Technology Gandhinagar, Gujarat, 382007 India\\
Email: krishna\_gopal@daiict.ac.in, mankg@computer.org}
}


%


\maketitle

\begin{abstract}
In this paper, we consider heterogeneous distributed storage systems (DSSs) having flexible reconstruction degree, where each node in the system has dynamic repair bandwidth and dynamic storage capacity. In particular, a data collector can reconstruct the file at time $t$ using some arbitrary nodes in the system and for an arbitrary node failure the system can be repaired by some set of arbitrary nodes. 
Using $min$-$cut$ bound, we investigate the fundamental tradeoff between storage and repair cost for our model of heterogeneous DSS. In particular, the problem is formulated as bi-objective optimization linear programing problem. For an arbitrary DSS, it is shown that the calculated $min$-$cut$ bound is tight. 
\end{abstract}


%
\IEEEpeerreviewmaketitle

\section{Introduction} 
 Cloud storage is a distributed storage system (DSS) in which information is stored on distinct nodes as encoded packets in a redundant manner. One can retrieve the file by contacting certain nodes in the system. In case of node failure, it can be repaired using other nodes in the system. For such DSSs, one has to optimize various parameters in the system such as storage capacity, repair bandwidth, availability, reliability, security and scalability. Such DSSs are used by many commercial systems like Facebook, Yahoo, IBM, Amazon and Microsoft Windows Azure system\cite{XorbasVLDB,Huang:2012:ECW:2342821.2342823,skydrive,amazonec2}.

In homogeneous DSSs (where each node has same storage capacity and same repair degree) \cite{dgwr7}, encoded data packets of a file with size $B$ are distributed among $n$ nodes (each having storage capacity $\alpha$) such that connecting any $k(<n)$ nodes, one can retrieve the whole file. In the case of any arbitrary node failure, system is repaired by downloading $\beta$ packets from any $d (<n)$ nodes, called helper nodes \cite{dgwr7}. In these systems, one can provide reliability by simply replicating or encoding the massage data packets. In the case of simple replication, storage minimization is inefficient. On the other hand, encoding of data packets using erasure MDS (maximum distance separable) codes leads to inefficiecy for bandwidth minimization during node repair process.
To optimize these conflicting parameters, in a seminal work Dimakis et. al \cite{dgwr7} introduced regenerating codes.
 In \cite{5550492,capacity}, tradeoff between storage capacity $\alpha$ and repair bandwidth $d\beta$ is analyzed by plotting tradeoff curve for regenerating code. All points on the tradeoff curve can be obtained by linear network codes over finite fields \cite{5206008, journals/jsac/Wu10}. In the tradeoff curve, by minimizing both parameters in different order, Minimum Bandwidth Regenerating (MBR) codes and Minimum Storage Regenerating (MSR) codes are obtained \cite{5550492}. Tradeoff between storage and repair bandwidth for exact-repair is studied in \cite{6874990}. In \cite{5513353}, Shah et al calculated cut-set lower bound on repair bandwidth for a special flexible setting  for homogeneous DSS. In a nice survey \cite{DARKYC11}, an overview of some existing results and repair models on DSS are explored. Recently in \cite{2015arXiv150103983P}, the tradeoff between storage capacity and repair bandwidth is investigated for exact repair linear regenerating codes for $k=d=n-1$. 

\begin{figure}
\centering
\includegraphics[scale=0.25]{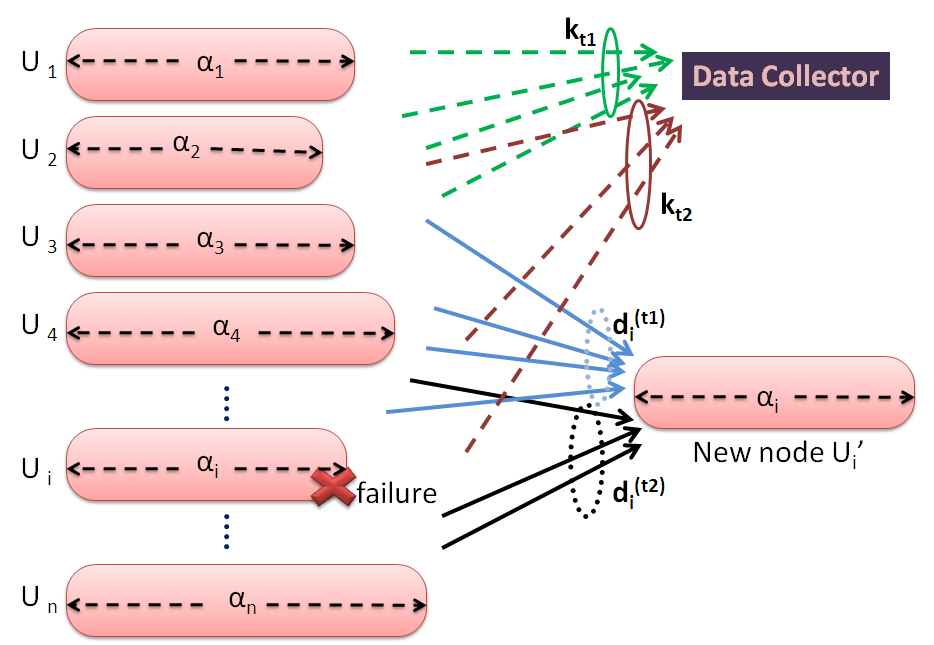}
\caption{A model of considered heterogeneous DSS is given here. In the heterogeneous DSS each node has flexible storage capacity $\alpha_i$ $(i\in\{1,2,\ldots,n\})$ and repair bandwidth. In the system, at time $t$ (in particular, $t\in\{t_1, t_2\}$) flexible reconstruction degree for a data collector, is $k_t$. Repair degree for an arbitrary node failure is also dynamic with respect to time. At time $t$, a failed node $U_i$ is repaired by some $d_i^{(t)}$ nodes.}
\label{example model}
\end{figure}
Heterogeneous DSSs are more close to real world scenarios where characterization of all storage nodes in various aspects are not necessarily uniform due to geographical environment and storage devices cost etc. 

 Many such heterogeneous DSS have been studied recently \cite{Kubiatowicz00oceanstore:an,DBLP:journals/ett/BianchiM96,DBLP:conf/acssc/PawarRZLR11}. In \cite{DBLP:journals/corr/abs-1202-1596,6736541, 6170563}, storage allocation problem is investigated to maximizes the probability of successful recovery. 
 For heterogeneous DSS, \cite{6033777} proved that repair cost can be reduced by allowing helper nodes to encode the codewords of other nodes. In \cite{Akhlaghi20102105}, Akhlaghi et al investigated the tradeoff between storage capacity and repair bandwidth for the generalized regenerating codes and shown that each point on curve is achievable. In the generalized regenerating code, set of all nodes is divided into two partitions. Every node in each partition has uniform parameters $(\alpha_i, d_i, \beta_i)$ ($\forall i\in\{1,2\}$) \cite{Akhlaghi20102105}. 
In ~\cite{6620426}, Ernvall et al calculated the capacity bounds of a heterogeneous DSS having dynamic repair bandwidth. The tradeoff curve is explored for non-homogeneous two rack model of DSS in \cite{6620424}. In \cite{DBLP:journals/corr/BenerjeeG14} capacity bound is calculated for heterogeneous DSSs with dynamic repair bandwidth, where node repair is done by some specific helper nodes. For the heterogeneous DSS, the storage node capacity depends on the repair bandwidth of each rack. In \cite{ETT:ETT2887}, tradeoff between system storage cost and system repair cost is investigated for heterogeneous DSSs with dynamic storage and repair cost.

\begin{figure}
\centering
\includegraphics[scale=0.27]{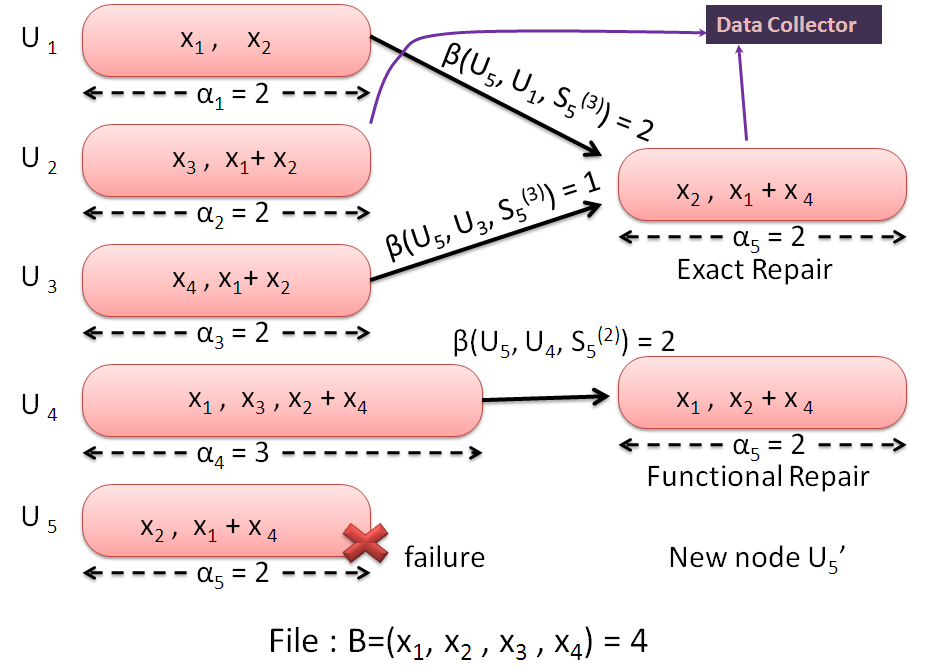}
\caption{A file with size $4$ units ($=B$) is divided into $11$ encoded packets on field $\F_q$. These packets are distributed among $6(=n)$ nodes in such a way that any data collector can download whole file by contacting at most $3(=k)$ nodes. In this heterogeneous DSS, $\vec{\alpha}=(\alpha_1, \alpha_2, \alpha_3, \alpha_4, \alpha_5)=(2\ 2\ 2\ 3\ 2)$. A failed node can be repaired by at most $2 (=d)$ nodes. Functional and exact repairs are shown for the node failure of $U_5$ with the help of surviving sets $S_5^{(2)}=\{U_4\}$ and $S_5^{(3)}=\{U_1, U_3\}$. Surviving set $S_5^{(3)}$ is not considered in Table \ref{example table} since $S_5^{(1)}\subsetneqq S_5^{(3)}$.}
\label{example}
\end{figure}

In this work, we consider a heterogeneous DSS where a file of size $B$ is distributed among $n$ nodes each with different storage capacities. File reconstruction is done 
in a flexible manner, where at any time instant $t$, data collector can reconstruct the file by connecting  some $k_t$ number of nodes. Hence the reconstruction degree $k_t$ for a file is flexible with respect to time and the number of nodes. On the other hand, in case of a node failure $U_i(1\leq i\leq n)$, it can be repaired at time $t$ by downloading packets from $d_i^{(t)}$ number of some nodes. Hence the repair degree $d_i^{(t)}$ is also flexible with respect to time and the number of nodes. Repair of failed node can be done in two ways, exact repair and functional repair. If the recovered packets in repair process is exact copy of lost packets then it is called \textit{exact repair}. On the other hand, if the recovered packets is some function of lost packets then the repair is \textit{functional repair}. The model of such heterogeneous DSS is shown in Figure \ref{example model}. A data collector reconstructs a distributed file by connecting $k_{t_j}(j\in\{1,2\})$ number of nodes at time $t_j$. In addition, a failed node $U_i$ repairs by $d_i^{(t_j)}$ number of some nodes at time $t_j$. 

An example of such heterogeneous DSS is considered in Figure \ref{example}. In this system, a file $B$ is divided into $4$ massage information packets $x_1$, $x_2$, $x_3$ and $x_4$. The massage information packets are encoded into $11$ packets by taking linear combination of massage information packets as $y_1 = x_1$, $y_2 = x_2$, $y_3 = x_3$, $y_4 = x_1+x_2$, $y_5 = x_4$, $y_6 = x_1+x_2$, $y_7 = x_1$, $y_8 = x_3$, $y_9 = x_2+x_4$, $y_{10} = x_2$ and $y_{11} = x_1+x_4$. The encoded packets $y_m(m\in[11])$ are distributed on the $5$ nodes such that packets $y_1$ and $y_2$ are stored on node $U_1$, packets $y_3$ and $y_4$ are distributed on node $U_2$, packets $y_5$ and $y_6$ are on node $U_3$, packets $y_7$, $y_8$ and $y_9$ are on node $U_4$ and remaining two packets are on node $U_5$. Clearly the storage node capacity $\alpha_i$ = $2(i\in[5]\backslash\{4\})$ and $\alpha_4$ = $3$. In this example, if node $U_5$ fails then it can be repaired by downloading packets $y_7$ and $y_9$ from node $U_4$. Since the recovered packets are function of lost packets so it is functional repair. On the other hand, node $U_5$ can be repaired exactly by downloading packets $y_1$, $y_2$ and $y_5$ from nodes $U_1$ and $U_3$ and solving $y_{10}=y_2$ and $y_{11}=y_1 + y_5$.


\textit{Contribution}: In this paper, we have calculated $min$-$cut$ bound for the considered heterogeneous DSS. 
For such heterogeneous DSS, we have established a bi-objective optimization linear programing problem subject to $min$-$cut$ bound. The solutions of the LP problem are plotted as a tradeoff curve between system storage and repair cost. In a heterogeneous DSS, system storage cost and system repair cost are average costs to store and repair unit information data on a node respectively. We have plotted some tradeoff curve and compared it with tradeoff for heterogeneous DSS as considered in \cite{ETT:ETT2887} and homogeneous DSS as investigated in \cite{capacity}. Some specific cases are investigated for the established bi-objective optimization problem.

\textit{Organization}: The paper is organized as follows. Section 2 collects the required preliminary concepts and describes our model. Section 3 investigates the $min$-$cut$ bound for our model. Under the constraints of the $min$-$cut$ bound, we also establish a bi-objective linear optimization problem to plot the tradeoff curve between storage and repair costs per node. Finally Section 4 concludes the paper with general remarks.
\section{Preliminaries}\label{2}    
In this paper, we focus our attention to heterogeneous DSS with parameters ($n,k,d$), where file is distributed among $n$ nodes, $k=\max_t\{k_t\}$ is the maximum reconstruction degree for the file, $d_i=\max_t\{d_i^{(t)}\}$ is the \textit{maximum repair degree} for a node $U_i$ at all time and $d=\max_i\{d_i\}$ is the \textit{maximum repair degree} among all nodes at any time. 
For each time $t$, one can define a \textit{reconstruction set} $\mathcal{A}_t$ as collections of the nodes having sufficient packets to reconstruct the file $i.e.$ $\mathcal{A}_t=\left\{U_{t_1}, U_{t_2},\ldots,U_{t_{k_t}}\right\}$. Clearly $|\mathcal{A}_t|=k_t$ and intersection of any two reconstruction set may be non-empty. Define $\mathscr{A}=\left\{\mathcal{A}_1,\mathcal{A}_2,\ldots,\mathcal{A}_t,\ldots\right\}$ as a set of all reconstruction sets. Note that the set $\mathscr{A}$ will be finite if all reconstruction sets $\mathcal{A}\in\mathscr{A}$ are distinct. Hence $\exists\ \omega\in\N$ such that $|\mathscr{A}|=\omega$.
For the considered example in Figure \ref{example}, $\mathscr{A}=\{\mathcal{A}_i:\forall i\in[7]\}$, where $\mathcal{A}_1=\{U_1, U_2, U_3\}$, $\mathcal{A}_2=\{U_1, U_3, U_5\}$, $\mathcal{A}_3=\{U_1, U_4\}$, $\mathcal{A}_4=\{U_2, U_4\}$, $\mathcal{A}_5=\{U_2, U_5\}$, $\mathcal{A}_6=\{U_3, U_4\}$ and $\mathcal{A}_7=\{U_4, U_5\}$. 
In the heterogeneous DSS, at time $t$, if a node $U_i\ (i \in [n])$ fails then certain nodes called helper nodes, download required packets and generate a new node say $U_i'$. The new node $U_i'$ replaces the failed node $U_i$ and the system is repaired. In particular, set of those helper nodes are called \textit{surviving set}. For a node $U_i$, let the number of distinct surviving sets are $\tau_i$. At the time instant $t$, indexing the surviving set by $\ell$, one can denote them by $S_i^{(\ell)}\stackrel{\triangle}{=} \left\{U_j:\mbox{ some }j\in\{1,2,...,n\}\backslash \{i\}\right\}$, where $\ell\in[\tau_i]$ \cite{DBLP:journals/corr/BenerjeeG14}. If a node failure $U_i$ repairs by nodes of surviving set $S_i^{(\ell)}$ then \textit{repair degree at the time instant $t$}, is $d_i^{(t)}=|S_i^{(\ell)}|$. Surviving sets for the heterogeneous DSS considered in Figure \ref{example} is listed in Table \ref{example table}. In this example, one can see that if a node $U_4$ fails then it can be repaired by connecting nodes $U_2$ and $U_3$ or nodes $U_2$ and $U_5$. Hence surviving sets for the node $U_4$ are $S_4^{(1)}$ and $S_4^{(2)}$. In Table \ref{example table}, for a given $i(i\in[5])$, $|S_i^{(\ell)}|$ is identical for all $\ell\in[\tau_i]$. In general, it may not be true. 
Also note that, in the table, we have chosen those surviving sets which are not the super set of other surviving set for the same node failure. In particular, the condition ensures the active participation of each node of an arbitrary surviving set during system repair process.
\begin{table}[ht]
\caption{Surviving sets for nodes in DSS as considered in Figure \ref{example}.}
\centering 
\begin{tabular}{|c||c|c|}
\hline
Nodes&Surviving sets                                                                        &  \# sets    \\[0.5ex]                                     
$U_i$& $S_i^{(\ell )}$                                                                      &   $\tau_i$  \\
\hline\hline
$U_1$&$S_1^{(1)}=\{U_2, U_4\}, S_1^{(2)}=\{U_2, U_5\}, $                                    &   5          \\
     &$S_1^{(3)}=\{U_3, U_4\}, S_1^{(4)}=\{U_3, U_5\}$, $S_1^{(5)}=\{U_4, U_5\}$.           &              \\
\hline 
$U_2$&$S_2^{(1)} = \{U_1, U_4\}, S_2^{(2)} =  \{U_3, U_4\},  S_2^{(3)} = \{U_4, U_5\}.$     &   3          \\
\hline 
$U_3$&$S_3^{(1)} = \{U_4\}, S_3^{(2)} =  \{U_5\}$.                                          &   2          \\
\hline 
$U_4$&$ S_4^{(1)} =  \{U_2, U_3\}, S_4^{(2)} =  \{U_2, U_5\}.$                              &   2          \\
\hline 
$U_5$&$ S_5^{(1)} =  \{U_3\}, S_5^{(2)} =  \{U_4\}.$                                        &   2          \\
\hline
\end{tabular}
\label{example table}
\end{table}
 
In brief, for a failed node $U_i$, if system is repaired by nodes of specific surviving set $S_i^{(\ell)}$ then the number of information packets downloaded by node $U_j\in S_i^{(\ell)}$ will be given by $\beta\left(U_i, U_j, S_i^{(\ell)}\right)>0$. For example, in Figure \ref{example}, all two packets from node $U_5$ and packet $y_3 = x_3$ from node $U_2$ is downloaded to repair node failure $U_4$. Note that $U_2,U_5\in S_4^{(2)}$. Hence $\beta\left(U_4,U_5,S_4^{(2)}\right)$ = $2$ and $\beta\left(U_4,U_2,S_4^{(2)}\right)$ = $1$.

If a failed node $U_i$ $(i\in[n])$ is repaired by nodes of surviving set $S_i^{(\ell)}$ then \textit{repair bandwidth} (denoted by $\gamma\left(U_i, S_i^{(\ell)}\right)$) for the node $U_i$ is the total number of packets downloaded by every nodes of the surviving set $S_i^{(\ell)}$. Mathematically


\begin{equation}
\gamma\left(U_i, S_i^{(\ell)}\right)=\sum_{\stackrel{j}{\mbox{\small such that }U_j\in S_i^{(\ell)}}}\beta\left(U_i, U_j, S_i^{(\ell)}\right).
\label{repair bandwidth}
\end{equation}
For example, in Figure \ref{example}, if node $U_5$ fails and it is repaired by nodes of surviving set $S_5^{(3)}$ (not considered in Table \ref{example table} since $S_5^{(1)}\subsetneqq S_5^{(3)}$) then 
\newline $\gamma\left(U_5, S_5^{(2)}\right)$= $\beta\left(U_5,U_1,S_5^{(3)}\right)$ + $\beta\left(U_5,U_3,S_5^{(3)}\right)$ = $2 + 1 =3$ units.

\begin{remark}
In this paper, at time instant $t$, single node failure is considered because simultaneously multi-node failures can be assumed as a sequence of single node failure. 
\end{remark}

\begin{remark}
One can find the tradeoff curve between repair cost and storage cost by optimization Problem \ref{equation} for the exact or functional repair using the surviving sets as the collection of those helper nodes which repair failed nodes as exact or functional respectively. 
\end{remark}

\begin{remark}
One can modify our heterogeneous DSS model by allowing some data collectors to reconstruct file separately at same time instant $t$ with flexible reconstruction degree each. For the particular modified model, the tradeoff curve between repair and storage cost can be plotted using optimization Problem \ref{equation}, if any two data collectors are not connected with some common node.
\end{remark}

To plot the tradeoff curve between storage capacity $\alpha$ and repair bandwidth $d\beta$ in a homogeneous DSS, Wu et al \cite{capacity} solved an optimization problem with constraint of $min$-$cut$ bound between the parameters. The bound is calculated by analyzing the information flow graph for the homogeneous DSS \cite{capacity}. In the similar manner, one can plot the trade off curve for our model. We consider the information flow graph (acyclic weighted directed graph $\mathcal{G}=(\mathcal{V}, \mathcal{E})$) \cite{ETT:ETT2887, DBLP:journals/corr/BenerjeeG14} for heterogeneous DSS as described in Figure \ref{flow graph}.

 For a heterogeneous DSS, at time $t$, the information flow graph $\mathcal{G}$ as shown in Figure \ref{flow graph}, is divided into $k_t+3$ ($k_t$ being flexible reconstruction degree for data collector at time $t$) steps, starting from step label $-1$ to label $k_t+1$. Step label $-1$ contains source node say $``s"$ and step label $k_t+1$ contains data collector node say $``D"$. A typical node $U_{\lambda_i}$ $(\forall i\in [n])$ in heterogeneous DSS, is mapped to a pair of vertices $``In_{\lambda_i}"$ and $``Out_{\lambda_i}"$ in $\mathcal{V}$ $s.t.$ $(In_{\lambda_i}, Out_{\lambda_i})\in\mathcal{E}$, where $\lambda_i$ is permute index on nodes. Storage capacity $\alpha_{\lambda_i}$ of node $U_{\lambda_i}$ is mapped to $w(In_{\lambda_i}, Out_{\lambda_i})$, where $w(In_{\lambda_i}, Out_{\lambda_i})$ is weight associated with edge $(In_{\lambda_i}, Out_{\lambda_i})\in\mathcal{E}$. In graph $\mathcal{G}$ as given in Figure \ref{flow graph}, at step label $0$, there are $2n$ number of vertices named $In_{\lambda_i}$ and $Out_{\lambda_i}$ associated with nodes $U_{\lambda_i}$ ($i\in[n]$) in heterogeneous DSS. 

\begin{figure}
\centering
\includegraphics[scale=0.27]{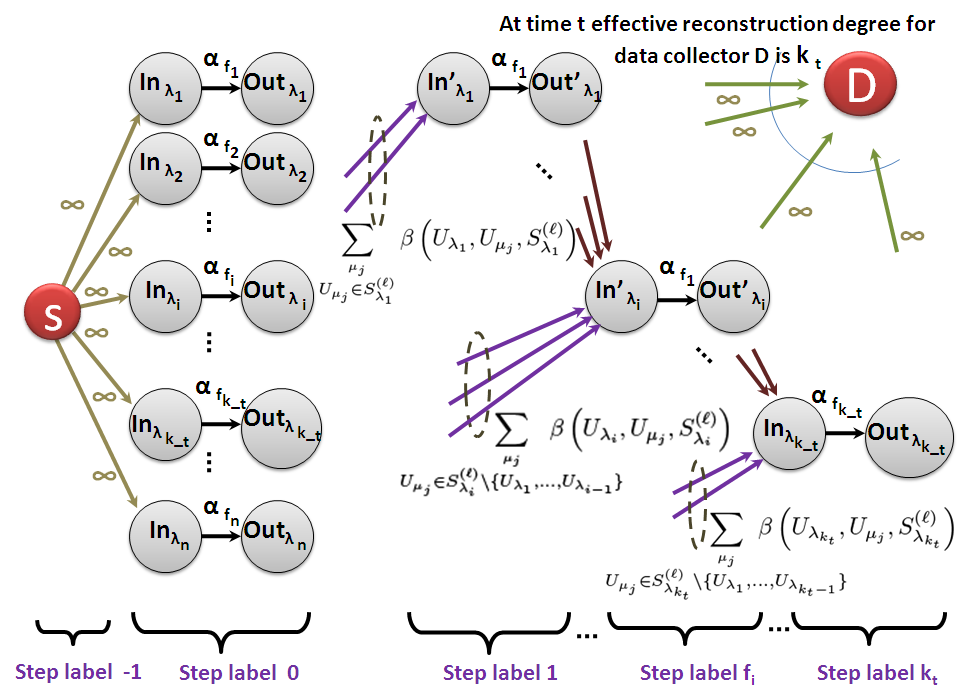}
\caption{Information flow graph $\mathcal{G}=(\mathcal{V}, \mathcal{E})$ for a heterogeneous DSS. The graph is divided into $k_t+3$ ($k_t$ is flexible reconstruction degree at time $t$ associated with data collector node ``$D$") step labels. Node $U_{\lambda_i}$ $(\forall i\in[n])$ in heterogeneous DSS is represented by a pair of nodes $In_{\lambda_i}$ and $Out_{\lambda_i}$. Sours nose ``$s$" is in step label $-1$, data collector node ``$D$" is in label $k_t+1$. Step label $0$ has $n$ pairs of nodes $In_{\lambda_i}$ and $Out_{\lambda_i}$. Each step from label $1$ to $k_t$ has one pair of nodes $In_{\lambda_j}$ and $Out_{\lambda_j}$ ($\forall j\in[k_t]$). Mansion that node ``$D$" (data collector node) at right upper corner is in step label $k_t+1$, in place of step label $k_t$.}
\label{flow graph}
\end{figure}

	 A failed node $U_{\lambda_i}$ ($i\in[n]$) in heterogeneous DSS, is repaired by generating new node $U_{\lambda_i}'$. The node $U_{\lambda_i}'$ is mapped to a new pair of nodes $In'_{\lambda_i}$ and $Out'_{\lambda_i}$  $s.t.$ $(In'_{\lambda_i}, Out'_{\lambda_i})\in\mathcal{E}$ with $w(In'_{\lambda_i}, Out'_{\lambda_i}) = \alpha_{\lambda_i}$. Every step label $j\in[k_t]$ contains one pair of nodes $In'_{\lambda_j}$ and $Out'_{\lambda_j}$. As shown in Figure \ref{flow graph}, in the heterogeneous DSS, system is repaired for the node failure $U_{\lambda_j}$ by downloading $\beta\left(U_{\lambda_j}, U_{\mu_p}, S_{\lambda_j}^{(\ell)}\right)$ amount of data from every node $U_{\mu_p}\left(\in S_{\lambda_j}^{(\ell)}=\{U_{\mu_p}: j\in[d_{\lambda_j}^{(t)}] \}\right)$, where $\mu_p$ is some permutation on nodes. For the particular system repair, each downloading process maps by one distinct edge from some previous step label to step label $j$ $s.t.$ $(Out'_{\mu_p}, In'_{\lambda_j})\in \mathcal{E}$ with $w(Out'_{\mu_p}, In'_{\lambda_j}) = \beta\left(U_{\lambda_j}, U_{\mu_p}, S_{\lambda_j}^{(\ell)}\right)$.
	In particular, if node $Out'_{\mu_p}$ does not exist then consider $Out_{\mu_p}$ from step label $0$ $s.t.$ $(Out_{\mu_j}, In'_{\lambda_j})\in\mathcal{E}$. In graph $\mathcal{G}$ exactly one node failure is considered in each step label.
	
	A data collector $D$ connects $k_t$ number of nodes of $\mathcal{A}_t$ = $\{U'_{\lambda_1}$, $U'_{\lambda_2}$,...,$U'_{\lambda_j}$,...,$U'_{\lambda_{k_t}}\}$. In information graph as in Figure \ref{flow graph}, data collector $D$ connects nodes $Out'_{\lambda_j} (\forall j\in[k_t])$ from step label $1$ to step label $k_t$ and downloads certain data file then $(Out'_{\lambda_j}, D)\in\mathcal{E}$ such that $w(Out'_{\lambda_j}, D)\rightarrow\infty$. 

\begin{example}
At time instant $t$, for the heterogeneous DSS in Figure \ref{example}, an example of information flow graph is shown in Figure \ref{flow graph example}. In particular, a data collector is connected with the nodes of $\mathcal{A}_1$ = $\{U_1$, $U_2$, $U_3\}$. In the information flow graph, if the nodes are failed then it will be  repaired by nodes of $S_1^{(1)}, S_2^{(1)}$ and $S_3^{(1)}$ respectively.
\end{example}
	
\begin{figure}
\centering
\includegraphics[scale=0.26]{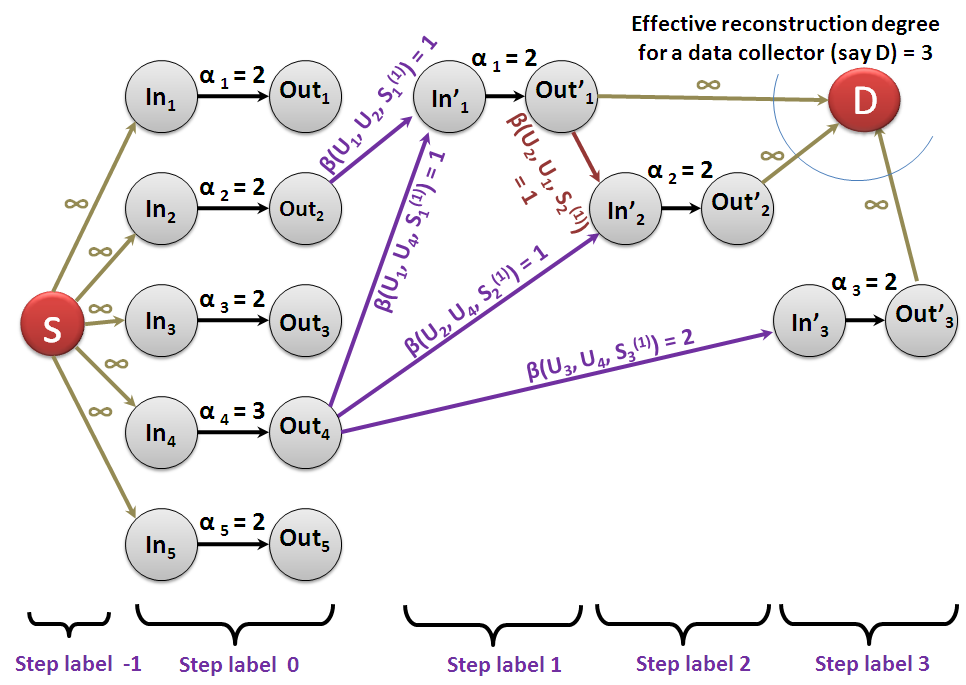}
\caption{For a heterogeneous DSS as considered in Figure \ref{example}, a information flow graph is shown for a specifics data collector connects with the nodes of $\mathcal{A}_1$=$\{U_1, U_2, U_3\}$. The particular information flow graph is plotted for surviving sequence $\left\langle S_1^{(1)}, S_2^{(1)}, S_3^{(1)}\right\rangle\in \mathscr{S}\left( \left\langle U_1, U_2, U_3\right\rangle \right)$. Mansion that node ``$D$" at right upper corner is in step label $4$, in place of step label $3$.}
\label{flow graph example}
\end{figure}
	

In \cite{ETT:ETT2887, capacity, DBLP:journals/corr/BenerjeeG14, 6620426}, $min$-$cut$ bound is calculated by analyzing flow passes through source node $s$ to data collector node $D$ across the information flow graph for a DSS. In the similar manner flow analysis is done for the model considered in this paper. Hence one can define flow across the information flow graph as follows.

\begin{definition}
A function $f:\mathcal{E}\rightarrow [0,\infty)$ is called flow on a information flow graph $\mathcal{G}=(\mathcal{V}, \mathcal{E})$ if,
\begin{enumerate}
	\item (capacity constraint:) $\forall (x,y)\in\mathcal{E}$,	$f((x,y))\leq c((x,y))$, where $c((x,y))=w(x,y)$ and $c((x,y))$ is capacity of edge $(x,y)$.
	\item (flow conservation constraint:) $\forall y\in\mathcal{V}\backslash\{s,t\}$,
	\begin{equation*}
	\sum_{\stackrel{x}{(x,y)\in\mathcal{E}}}f((x,y))=\sum_{\stackrel{z}{(y,z)\in\mathcal{E}}}f((y,z)).
	\end{equation*}
\end{enumerate}
\end{definition}

For more details and example on flow function, cite \cite{Ahlswede00networkinformation, Elias:1956}.

For a given information flow graph $\mathcal{G}=(\mathcal{V}, \mathcal{E})$, value of flow delivered to a data collector node say $D$ is defined as total amount of flow passes through the edges $(x,D)\in\mathcal{E}$ for all possible $x\in\mathcal{V}$.
	 For networks, maximum possible value of flow delivered to $D$ is governed by \textit{min-cut max-flow} theorem \cite{Ahlswede00networkinformation, Elias:1956, Ford-Fulkerson_algo}.
 \textit{Min-cut max-flow} theorem says that across the network, maximum possible value of flow passes from source $s$ to specific data collector $D$ denoted by $max\mbox{-}flow(s,D)$, is equal to minimum $cut$-$capacity(s,D)$, where 
\begin{equation*}
\min cut\mbox{-}capacity(s,D)=\min_{\stackrel{cut(\mathcal{X},\overline{\mathcal{X}});}{\stackrel{s\in\mathcal{X},D\in\overline{\mathcal{X}};}{\mathcal{X}\cup\overline{\mathcal{X}}=\mathcal{V}.}}}\left\{cut\mbox{-}capacity(\mathcal{X},\overline{\mathcal{X}})\right\}.   
\end{equation*}
 Note that $cut(\mathcal{X},\overline{\mathcal{X}})$ represents the set of all edges having one end vertex in set $\mathcal{X}$ and other vertex in set $\overline{\mathcal{X}}$ such that removing those all edges will improve the number of components in graph $\mathcal{G}=(\mathcal{V},\mathcal{E})$. Here $cut$-$capacity(\mathcal{X},\overline{\mathcal{X}})$ is the sum of capacity of all edges in $cut(\mathcal{X},\overline{\mathcal{X}})$. At time $t$, for a specific data collector $D$ which connects nodes $U_{\lambda_i}$ of set $\mathcal{A}_t\in\mathcal{A}$, has $|\mathcal{A}_t|\,!\prod_{i=1}^{|\mathcal{A}_t|}\tau_{\lambda_i}$ number of distinct information flow graphs are exist. For every information flow graph $\mathcal{G}=(\mathcal{V}, \mathcal{E})$, $D$ can recover the whole file $B$ so 
\begin{equation*}
B\leq\min_{\mathcal{G}}\max\mbox{-}flow(s,D).
\end{equation*}
 By \textit{min-cut max-flow} theorem for an arbitrary data collector $D$ one can compute, 

\begin{equation*}
B\leq\min_t\min_{\mathcal{G}}\max\mbox{-}flow(s,D).
\end{equation*}

In \cite{capacity}, for an information flow graph, flow analysis is done by taking topological order of failed node connected with data collector. In this paper, we are defining some sequences of nodes and corresponding surviving sets for our model to analyze flow. The definitions are as follows.

\begin{definition}
A set of all possible sequences of nodes in a reconstruction set $\mathcal{A}_j \in \mathcal{A}$ is called reconstruction sequence set and denoted by $\mathscr{A}(\mathcal{A}_j)=\left\{\left\langle U_{\lambda_i}\right\rangle_{i=1}^{|\mathcal{A}_j|}:U_{\lambda_i}\in\mathcal{A}_j\right\}$, where $\left\langle U_{\lambda_i}\right\rangle_{i=1}^{|\mathcal{A}_j|}$ represents a sequence of distinct nodes of set $\mathcal{A}_j\in\mathcal{A}$. Clearly $|\mathscr{A}(\mathcal{A}_j)|$ = $|\mathcal{A}_j|\,!$.
\end{definition}
 For example, in Figure \ref{example}, $\mathscr{A}(\mathcal{A}_3)=\left\{\left\langle U_1, U_4\right\rangle, \left\langle U_4, U_1\right\rangle \right\}$ etc. 
\begin{definition}
For a reconstruction set $\mathcal{A}_j\in\mathcal{A}$, one can define sequences of surviving sets $S^{(\ell)}_{\lambda_i} (\forall i\in[|\mathcal{A}_j|], \exists \ell\in[\tau_{\lambda_i}])$ such that $U_{\lambda_i}\in\mathcal{A}_j$. Surviving sequence associated with node sequence $\left\langle U_{\lambda_i}\right\rangle_{i=1}^{|\mathcal{A}_j|}\in\mathscr{A}$ can be denoted by $\left\langle S^{(\ell)}_{\lambda_i}\right\rangle_{i=1}^{|\mathcal{A}_j|}$. 
\end{definition}
For example, in Figure \ref{example}, a possible surviving sequence for the node sequence $\left\langle U_1, U_4 \right\rangle$ is $\left\langle S_1^{(3)}, S_4^{(2)}\right\rangle$. 
\begin{definition}
Set of all surviving sequences associated with a node sequence $\left\langle U_{f_i}\right\rangle_{i=1}^{|\mathcal{A}_j|}$ can be defined as follows, 
\begin{equation*}
\mathscr{S}\left(\left\langle U_{\lambda_i}\right\rangle_{i=1}^{|\mathcal{A}_j|}\right)=\left\{\left\langle S^{(\ell)}_{\lambda_i}\right\rangle_{i=1}^{|\mathcal{A}_j|}:\exists\ell\in [\tau_{\lambda_i}]\right\}.
\end{equation*} 
 Clearly $\left|\mathscr{S}\left(\left\langle U_{\lambda_i}\right\rangle_{i=1}^{|\mathcal{A}_j|}\right)\right|$ = $\left(\prod_{i=1}^{|\mathcal{A}_j|}\tau_{\lambda_i}\right)\,!$.
\end{definition}
 For example, in Figure \ref{example}, one can see that $\mathscr{S}\left(\left\langle U_1, U_4 \right\rangle\right)=\left\{\left\langle S_1^{(\ell_1)}, S_4^{(\ell_2)}\right\rangle : \exists \ell_1\in[5], \exists \ell_2\in[2] \right\}$ etc.

In \cite{ETT:ETT2887}, Quan et al have given tradeoff curve between system storage cost and system repair cost for heterogeneous DSS with uniform reconstruction degree. Similarly one can give tradeoff curve between system storage cost and system repair cost for heterogeneous DSS model considered in our paper. For our model, we define system storage cost, node storage cost and system repair cost as follows.

\begin{definition} (System storage cost): Total amount of cost $C_s(\vec{\alpha})$ to store unit data in heterogeneous DSS($n, k, d$) is called system storage cost, where storage amount vector $\vec{\alpha}\triangleq(\alpha_1, \alpha_2,\ldots,\alpha_n)$, storage cost vector $\vec{s}\triangleq(s_1, s_2,\ldots, s_n)$, $\alpha_i$ is storage capacity of node $U_i$ and $s_i$ is the cost to store unit information data in node $U_i$ $(\forall i\in[n])$. Clearly
\begin{equation*}
C_s(\vec{\alpha}) = \frac{1}{B}\sum_{j=1}^ns_j\alpha_j
\end{equation*}
\end{definition}
System storage cost $C_s(\vec{\alpha})$ for the example considered in Figure \ref{example} with $\vec{s} = (100, 10, 10, 10, 1)$ is $68$ cost units.


\begin{definition} (Node repair cost): The average amount of cost to repair a node $U_i (i\in[n])$ in heterogeneous DSS($n, k, d$) is called node repair cost $r(\beta_i)$ associated with repair cost vector $\vec{r}\triangleq(r_1, r_2,\ldots, r_n)$ $s.t.$
\begin{equation}
r(\beta_i) = \frac{1}{B\tau_i}\sum_{\ell=1}^{\tau_i}\sum_{\stackrel{j}{ U_j\in S_i^{(\ell)}}}r_j\beta(U_i, U_j, S_i^{(\ell)}),
\end{equation} where $r_j$ is cost to download unit amount of data from node $U_j$ during repair process. Clearly node repair vector $r(\vec{\beta}) \triangleq (r(\beta_1), r(\beta_2),\ldots, r(\beta_n))$.
\end{definition}
In the example considered in Figure \ref{example}, if $\vec{r} = (10, 1, 1, 1, 1)$ then node repair cost vector $r(\vec{\beta})=(r(\beta_1), r(\beta_2), r(\beta_3), r(\beta_4), r(\beta_5))$ = $(\frac{1}{2}, \frac{4}{3}, \frac{1}{2}, \frac{3}{4}, \frac{1}{2})$.



\begin{definition} (System repair cost): System repair cost is total amount of cost to repair all nodes in heterogeneous DSS($n, k, d$) and denoted by $C_r(\vec{\beta})$. Mathematically 
\begin{equation*}
C_r(\vec{\beta}) = \sum_{j=1}^{n}r(\beta_j).
\end{equation*}
\end{definition} 

Clearly $C_r(\vec{\beta})=\frac{43}{12}$ cost unit for $\vec{r} = (10, 1, 1, 1, 1)$ in the example considered in Figure \ref{example}.

Now we give some results and analysis for the $min$-$cut$ bound for the model of heterogeneous DSS considered in this paper in the next section. 

\section{Results}\label{3}
For our model, it is shown that minimum possible value of flexible reconstruction degree is lower bound of cardinality of any \textit{cut set} which separates source node and data collector node. For the heterogeneous DSS \textit{min cut} bound is calculated in Theorem \ref{min cut bound theorem}. Using that \textit{min cut} bound, it is shown that file size should be lower bound of \textit{min cut} bound for the heterogeneous DSS. Using the particular bound as constraint, a bi-objective optimization linear programing problem is formulated to minimize system storage cost and system repair cost for the considered heterogeneous model. A family of solutions is calculated for the optimization problem by substituting some numerical values of system parameters. The numerical parameter is plotted the tradeoff curve between system storage cost and system repair cost. The curve is compared with tradeoff curve for homogeneous DSS \cite{capacity} and tradeoff curve for heterogeneous DSS \cite{ETT:ETT2887}.


\begin{lemma} 
An arbitrary information flow graph $\mathcal{G}=(\mathcal{V}, \mathcal{E})$ with source node $s\in\mathcal{X}$, flexible reconstruction degree $k_t$ associated with data collector node $D\in\overline{\mathcal{X}}$ 
has \begin{equation*}\min_{\mathcal{X}\subset\mathcal{V}}\left\{|cut(\mathcal{X},\overline{\mathcal{X}})|:cut(\mathcal{X},\overline{\mathcal{X}})\neq\phi\right\}\geq \min_{t}\{k_t\},\end{equation*} where $\mathcal{X}\cup\overline{\mathcal{X}}=\mathcal{V}$.
\end{lemma}
\begin{proof}
Consider a heterogeneous DSS associated with some information flow graphs. For any arbitrary information flow graph $\mathcal{G}=(\mathcal{V}, \mathcal{E})$, $\exists$ $ \mathcal{X}\subset\mathcal{V}$ such that $s\in\mathcal{X}$, $D\in\overline{\mathcal{X}}$. Since information flow graph is connected graph so $cut(\mathcal{X},\overline{\mathcal{X}})\neq\phi$ for any nonempty set $\mathcal{X}$ and $\overline{\mathcal{X}}$. To retrieve the distributed file for the case $\overline{\mathcal{X}}=\{D\}$, one has to connect at least $\min_t\{k_t\}$ number of nodes among $n$ nodes. Hence each node of an arbitrary set of $\min_t\{k_t\}$ number of nodes, has some encoded data of distinct part of massage data. So there are at least $\min_{t}\{k_t\}$ number of edges having end point as node $D$. Mathematically, $|cut(\mathcal{X},\overline{\mathcal{X}})| = \min_{t}\{k_t\}$ for $\overline{\mathcal{X}}=\{D\}$. Again, if $|\overline{\mathcal{X}}|>1$ then some edges in $cut(\mathcal{X},\overline{\mathcal{X}})$ represents downloading process for system repair. In particular, a node failure among the $\min_t\{k_t\}$ number of the specific nodes, can not be repaired by the some subset of the remaining $\min_{t}\{k_t\}-1$ number of the nodes. Reason behind that, each node in the set of $\min_{t}\{k_t\}$ number of nodes has encoded data packets of some unique message data packets. Hence, there must exist some helper nodes other then the $\min_t\{k_t\}-1$ number of nodes for the repair the failed node. So $|cut(\mathcal{X},\overline{\mathcal{X}})|\geq \min_{t}\{k_t\}$. But $\mathcal{X}$ is any arbitrary nonempty subset such that $cut(\mathcal{X},\overline{\mathcal{X}})$ exist, so $|cut(\mathcal{X},\overline{\mathcal{X}})|\geq \min_{t}\{k_t\}$ for all possible $cut(\mathcal{X},\overline{\mathcal{X}})\neq\phi$. This proves the lemma. 
\end{proof}




In a heterogeneous DSS, information delivered to data collector $D$ depends on min\textit{cut-capacity($s, D$)}. The Theorem \ref{min cut bound theorem} gives the lower bound of $\min cut\mbox{-}capacity(s,D)$. 

\begin{theorem}
(min-cut bound) 
For a given heterogeneous DSS with an arbitrary data collector $D$ associated with flexible reconstruction degree $k_t$, the $\min cut\mbox{-}capacity(s,D)$ is bounded below by $\mathscr{Q}$ as given by Equation (\ref{Q}), $i.e.$ 
\begin{equation}
\min cut\mbox{-}capacity(s,D)\geq \mathscr{Q}.
\label{1}
\end{equation}
\label{min cut bound theorem}
\end{theorem}
\begin{proof}
Consider a heterogeneous DSS $(n, k, d)$ associated with some information flow graphs. Every information flow graph $\mathcal{G}=(\mathcal{V}, \mathcal{E})$ has a source node $s$, a data collector node $D$ associated with effective reconstruction degree $k_t$. In the heterogeneous DSS, a failed node $U_i$ can be repaired by nodes of some surviving set $S_i^{(\ell)}$, where $\ell\in [\tau_i]$.


Let $\mathcal{X}\subset\mathcal{V}$, $\mathcal{X}\cup\overline{\mathcal{X}}=\mathcal{V}$, $s\in\mathcal{X}$ and $D\in\overline{\mathcal{X}}$ such that some nonempty subset $cut(\mathcal{X}, \overline{\mathcal{X}})\subset\mathcal{E}$ exist. Now if $\mathcal{X}=\mathcal{V}\backslash \{D\}$ then $cut\mbox{-}capacity(\mathcal{X}, \overline{\mathcal{X}})\rightarrow\infty$. Similarly if $\mathcal{X}=\{s\}$ then again $cut\mbox{-}capacity(\mathcal{X}, \overline{\mathcal{X}})\rightarrow\infty$. Hence $\min cut\mbox{-}capacity(\mathcal{X}, \overline{\mathcal{X}})$ would be obtained by all those $Out_j'\in\overline{\mathcal{X}}$ and $In_i\in\mathcal{X}$ since it will give a finite $cut$-$capacity(\mathcal{X},\overline{\mathcal{X}})$, where $i\in[n]$ and $j\in[k_t]$.

Information flow graph $\mathcal{G}=(\mathcal{V}, \mathcal{E})$ is directed acyclic graph so it can be represented in a topological order of its vertices. For the topological order, sequences of node failure and corresponding sequence of surviving sets are arranged by using definitions as given in previous section. For that assume at time $t$, data collector $D$ connects with all nodes of a set $\mathcal{A}_t\in\mathcal{A}$ and reconstruct the file $B$. $\mathcal{A}(\mathcal{A}_t)$ is the set of all possible sequences of nodes of $\mathcal{A}_t\in\mathcal{A}$. A sequence $\left\langle U_{\lambda_i}\right\rangle_{i=1}^{k_t}\in\mathcal{A}(\mathcal{A}_t)$ represents the order of nodes failure of specific set $\mathcal{A}_t$. Recall the set of all possible surviving sequences $\left\langle S_{\lambda_i}^{(\ell)}\right\rangle_{i=1}^{k_t}$ associated with a node sequence $\left\langle U_{\lambda_i}\right\rangle_{i=1}^{k_t}$, is $\mathscr{S}\left(\left\langle U_{\lambda_i}\right\rangle_{i=1}^{k_t}\right)$.  

For a specific node sequence $\left\langle U_{\lambda_i}\right\rangle_{i=1}^{k_t}$ with a specific surviving sequence $\left\langle S^{(\ell)}_{\lambda_i}\right\rangle_{i=1}^{k_t}$, one can analyze the following

For $Out_{\lambda_1}'\in\overline{\mathcal{X}}$ associated with the first node in node sequence $\left\langle U_{\lambda_i}\right\rangle_{i=1}^{k_t}$ , the following two cases are possible.
\begin{itemize}
	\item If $In_{\lambda_1}'\in\mathcal{X}$ then edge $(In_{\lambda_1}', Out_{\lambda_1}')\in cut(\mathcal{X},\overline{\mathcal{X}})$. Hence $\alpha_{\lambda_1}$ will contribute in $cut$-$capacity(\mathcal{X},\overline{\mathcal{X}})$.
	\item If $In_{\lambda_1}'\in\overline{\mathcal{X}}$ then edges $(Out_{\mu_j}, In_{\lambda_1}')\in cut(\mathcal{X},\overline{\mathcal{X}})$, where $U_{\mu_j}\in S_{\lambda_1}^{(\ell)}$ and $S_{\lambda_1}^{(\ell)}\in\left\langle S^{(\ell)}_{\lambda_i}\right\rangle_{i=1}^{k_t}$ for any $\ell\in[\tau_{\lambda_1}]$. Hence this case contribute in $cut$-$capacity(\mathcal{X},\overline{\mathcal{X}})$ by
	\begin{equation*}
	\sum_{\stackrel{\mu_j}{U_{\mu_j}\in S_{\lambda_1}^{(\ell)}}}{\beta\left(U_{\lambda_1}, U_{\mu_j},S_{\lambda_1}^{(\ell)}\right)}.
	\end{equation*}
\end{itemize}
So contribution in $\min cut$-$capacity(\mathcal{X},\overline{\mathcal{X}})$ supported by node $U_{\lambda_1}$ is 
	\begin{equation*}
	\min \left\{\alpha_{\lambda_1}, \sum_{\stackrel{\mu_j}{U_{\mu_j}\in S_{\lambda_1}^{(\ell)}}}{\beta\left(U_{\lambda_1}, U_{\mu_j}, S_{\lambda_1}^{(\ell)}\right)}\right\}.
	\end{equation*}
	
	If a node $U_p(\forall p\in[n])$ fails in the system then all nodes of some surviving set $S_p^{(\ell)}$ will generate a new node $U'_p$ with same characteristic. At a time instant $t$, one of them is in the system. Hence for the remaining part of the proof, we are writing $U_p$ in place of $U'_p$.
	
	For the remaing part of the proof we have used the notation $U_p(\forall p\in[n])$ in place of $U_p'$ since characteristics of both nodes $U_p$ and $U_p'$ are same and one of them appears at instant.
	
In general to compute contribution in $\min cut$-$capacity(\mathcal{X},\overline{\mathcal{X}})$ supported by node $U_{\lambda_i}\in\left\langle U_{\lambda_i}\right\rangle_{i=1}^{k_t}$ assume $Out_{\lambda_i}'\in\overline{\mathcal{X}}$. Again following two cases are possible.
	
	\begin{itemize}
		\item If $In_{\lambda_i}'\in\mathcal{X}$, then edge $(In_{\lambda_i}', Out_{\lambda_i}')\in cut(\mathcal{X},\overline{\mathcal{X}})$. Hence $\alpha_{\lambda_i}$ will contribute in $cut$-$capacity(\mathcal{X},\overline{\mathcal{X}})$.
		\item If $In_{\lambda_i}'\in\overline{\mathcal{X}}$ then all possible edges $(Out_{\mu_j}, In_{\lambda_i}')$ $s.t.$ $U_{\mu_j}\in S_{\lambda_i}^{(\ell)}\backslash\{U_{\lambda_1}, U_{\lambda_2},\ldots, U_{\lambda_{i-1}}\}$ associated $S_{\lambda_i}^{(\ell)}\in\left\langle S^{(\ell)}_{\lambda_i}\right\rangle_{i=1}^{k_t}$ for any $\ell\in [\tau_{\lambda_i}]$, will contribute in $cut(\mathcal{X},\overline{\mathcal{X}})$. Edges $(Out_{\lambda_j}, In_{\lambda_i}')$ associated with node $U_{\mu_j}\in S_{\lambda_i}^{(\ell)}\backslash\{U_{\lambda_1}, U_{\lambda_2},\ldots, U_{\lambda_{i-1}}\}$ are newly investigated from step label $0$ for $cut(\mathcal{X},\overline{\mathcal{X}})$. Edges $(Out_{\lambda_m}', In_{\lambda_i}')$ must be excluded because they have investigated earlier at step label $m$, where $m\in[i-1]$ $s.t.$ $U_{\lambda_m}\in S_{\lambda_i}^{(\ell)}$.		
		Hence this case contribute in $cut$-$capacity(\mathcal{X},\overline{\mathcal{X}})$ by
	\begin{equation*}
	\sum_{\stackrel{\mu_j}{U_{\mu_j}\in S_{\lambda_i}^{(\ell)}\backslash\{U_{\lambda_1}, U_{\lambda_2},\ldots, U_{\lambda_{i-1}}\}}}{\beta\left(U_{\lambda_i}, U_{\mu_j},S_{\lambda_i}^{(\ell)}\right)}.
	\end{equation*}
	
	\end{itemize}
So contribution in $\min cut$-$capacity(\mathcal{X},\overline{\mathcal{X}})$ by node $U_{\lambda_i}$ is 
	\begin{equation*}
	\min \left\{\alpha_{\lambda_i}, \sum_{\stackrel{\mu_j}{U_{\mu_j}\in S_{\lambda_i}^{(\ell)}\backslash\{U_{\lambda_1}, U_{\lambda_2},\ldots,U_{\lambda_{i-1}}\}}}{\beta\left(U_{\lambda_i}, U_{\mu_j}, S_{\lambda_i}^{(\ell)}\right)}\right\}.
	\end{equation*} 
	At the time instant $t$, if data collector $D$ connects with each nodes $U_{\lambda_i}\in\mathcal{A}_t$, $i\in[k_t]$ then for a specific node sequence $\left\langle U_{\lambda_i}\right\rangle_{i=1}^{k_t}$ associated with a specific surviving sequence $\left\langle S_{\lambda_i}^{(\ell)}\right\rangle_{i=1}^{k_t}$ the contribution in $\min cut$-$capacity(\mathcal{X},\overline{\mathcal{X}})$ is 
		\begin{equation*}
	\sum_{i=1}^{k_t}\min \left\{\alpha_{\lambda_i}, \sum_{\stackrel{\mu_j}{U_{\mu_j}\in S_{\lambda_i}^{(\ell)}\backslash\{U_{\lambda_1}, U_{\lambda_2},\ldots,U_{\lambda_{i-1}}\}}}{\beta\left(U_{\lambda_i}, U_{\mu_j},S_{\lambda_i}^{(\ell)}\right)}\right\}.
	\end{equation*} 
		
Now one can find $\min cut$-$capacity(s,D)$ for a specific $D$ by taking minimum among all possible $cut$-$capacity(\mathcal{X},\overline{\mathcal{X}})$ which is calculated for all possible node sequences $\left\langle U_{\lambda_i}\right\rangle_{i=1}^{k_t}$ among all possible associated surviving sequences $\left\langle S_{\lambda_i}^{(\ell)}\right\rangle_{i=1}^{k_t}$. The min$cut$-$capacity(s,D)$ is the minimum value of the particular min$cut$-$capacity(s,D)$ calculated for all possible specific $D$. For a given heterogeneous DSS, associated with any arbitrary data collector $D$ one can find $\min cut$-$capacity(s,D)$ as Inequality (\ref{1}) by using Equation (\ref{equation 1}). The particular Equation (\ref{equation 1}) holds because index $\lambda_i$ of storage node capacity is governed by index $\lambda_i$ of nodes in node sequence $\left\langle U_{\lambda_i}\right\rangle_{i=1}^{k_t}$. 

\begin{figure*}
\begin{equation}
\begin{split}
\min_{\left\langle U_{\lambda_i}\right\rangle_{i=1}^{k_t}\in\mathscr{A}(\mathcal{A}_t)}&\min_{\left\langle S_{\lambda_i}^{(\ell)}\right\rangle_{i=1}^{k_t}\in\mathscr{S}\left(\left\langle U_{\lambda_i}\right\rangle_{i=1}^{k_t}\right)}\sum_{\stackrel{i=1}{U_{\lambda_i}\in\mathcal{A}_t}}^{k_t}\min\left\{\alpha_{\lambda_i}, \sum_{\stackrel{\mu_j}{U_{\mu_j}\in S_{\lambda_i}^{(\ell)}\backslash\{U_{\lambda_1},\ldots,U_{\lambda_{i-1}}\}}}{\beta\left(U_{\lambda_i}, U_{\mu_j}, S_{\lambda_i}^{(\ell)}\right)}\right\} \\
=&\min_{\left\langle U_{\lambda_i}\right\rangle_{i=1}^{k_t}\in\mathscr{A}(\mathcal{A}_t)}\sum_{\stackrel{i=1}{U_{\lambda_i}\in\mathcal{A}_t}}^{k_t}\min\left\{\alpha_{\lambda_i},\min_{\left\langle S_{\lambda_i}^{(\ell)}\right\rangle_{i=1}^{k_t}\in\mathscr{S}\left(\left\langle U_{\lambda_i}\right\rangle_{i=1}^{k_t}\right)}\sum_{\stackrel{\mu_j}{U_{\mu_j}\in S_{\lambda_i}^{(\ell)}\backslash\{U_{\lambda_1},\ldots,U_{\lambda_{i-1}}\}}}{\beta\left(U_{\lambda_i}, U_{\mu_j}, S_{\lambda_i}^{(\ell)}\right)}\right\}
\label{equation 1}
\end{split}
\end{equation}
\end{figure*}

One can easily observe that the bound is tight since the \textit{min-cut} bound is calculated by taking the minimum value of all possible \textit{cut} bounds. Hence one can say the following.

The $min$-$cut$ bound is calculated for all possible node sequences $\left\langle U_{\lambda_i}\right\rangle_{i=1}^{k_t}$ associated with all possible surviving sequences $\left\langle S_{\lambda_i}^{(\ell)}\right\rangle_{i=1}^{k_t}$. Hence there exist at least one surviving sequences, say, $\left\langle S_{\lambda_i}^{(\ell^*)}\right\rangle_{i=1}^{k_t}$ associated with node sequence, say, $\left\langle U_{\lambda_i}^*\right\rangle_{i=1}^{k_t}$ for which the inequality holds with equality $i.e.$ the $min$-$cut$ bound Inequality (\ref{1}) is tight.
\end{proof}

\begin{figure*}
\begin{equation}
\mathscr{Q}=\min_{\stackrel{\mathcal{A}_t\in\mathcal{A}}{|\mathcal{A}_t|=k_t}}\min_{\left\langle U_{\lambda_i}\right\rangle_{i=1}^{k_t}\in\mathscr{A}(\mathcal{A}_t)}\sum_{\stackrel{i=1}{U_{\lambda_i}\in\mathcal{A}_t}}^{k_t}\min\left\{\alpha_{\lambda_i},\min_{\left\langle S_{\lambda_i}^{(\ell)}\right\rangle_{i=1}^{k_t}\in\mathscr{S}\left(\left\langle U_{\lambda_i}\right\rangle_{i=1}^{k_t}\right)}\sum_{\stackrel{\mu_j}{U_{\mu_j}\in S_{\lambda_i}^{(\ell)}\backslash\{U_{\lambda_1},\ldots,U_{\lambda_{i-1}}\}}}{\beta\left(U_{\lambda_i}, U_{\mu_j}, S_{\lambda_i}^{(\ell)}\right)}\right\}
\label{Q}
\end{equation}
\end{figure*}

\begin{remark}
For a given heterogeneous DSS, at time $t$, if an arbitrary data collector connects each node $U_{\lambda_j}$ in subset $\mathcal{A}_t\in\mathcal{A}$ then total number of possible information flow graphs are given by
\begin{equation*}
\sum_{\stackrel{\mathcal{A}_t}{\mathcal{A}_t\in\mathcal{A}}}\left(|\mathcal{A}_t|\,!\prod_{j=1}^{|\mathcal{A}_t|}\tau_{\lambda_j}\right).
\end{equation*}
In particular, for a specific information flow graph, the total number of computational comparisons are $2^{|\mathcal{A}_t|}$. Hence One can say that the time complexity to calculate $min$-$cut$ bound is 
\begin{equation*}
\mathcal{O}\left(\sum_{\stackrel{\mathcal{A}_t}{\mathcal{A}_t\in\mathcal{A}}}\left(2^{|\mathcal{A}_t|}(|\mathcal{A}_t|\,!)\prod_{j=1}^{|\mathcal{A}_t|}\tau_{\lambda_j}\right)\right).
\end{equation*}
\end{remark}

By Theorem \ref{min cut bound theorem} one can calculate the minimum requirement of storage node capacity and repair bandwidth to store a file with size $B$. In other words the upper bound of stored file with size $B$ is given by the following lemma.
\begin{lemma}
If a file with size $B$ is stored in a given heterogeneous DSS $(n, k, d)$ then 
\begin{equation}
B\leq\mathscr{Q},
\label{reduced min cut bound}
\end{equation} 
where$\mathscr{Q}$ is given in Equation ($\ref{Q}$) and remaining used notations have common meaning as defined in previous sections.

\label{lemma}
\end{lemma}
\begin{proof}
Any arbitrary data collector node $D$ must be able to reconstruct the whole file with size $B$. Hence maximum information flow value delivered to any data collector, should be at least $B$. Now using $min$-$cut$ $max$-$flow$ theorem and Theorem \ref{min cut bound theorem} one can prove the lemma.
\end{proof}

\begin{example}
The $\min cut$-$capacity(s,D)$ for the information flow graph as shown in Figure \ref{flow graph example}, will be $\min\left\{\alpha_1, \beta\left(U_1, U_2, S_1^{(1)}\right)+\beta\left(U_1, U_4, S_1^{(1)}\right)\right\}$ + $\min\left\{\alpha_2, \beta\left(U_2, U_4, S_2^{(1)}\right)\right\}$ + $\min\left\{\alpha_3, \beta\left(U_3, U_4, S_3^{(1)}\right)\right\}$=$2+1+2=5$ units.
\end{example}

Now one can frame a optimization problem to find minimum system storage cost and system repair cost under the constraint that the maximum possible information deliver to data collector node $D$ is at lest $B$.  
\begin{problem}
\begin{equation*}
\begin{split}
& \mbox{Minimize: }[C_s(\vec{\alpha}), C_r(\vec{\beta})]\\
& \mbox{subject to}\\
&\hspace{17mm}\mbox{ Inequality } (\ref{reduced min cut bound}); \\
&\hspace{23mm}\alpha_i\geq 0;\\
&\hspace{8mm}\beta\left(U_i, U_j, S_{i}^{(\ell)}\right)\geq 0; 
\end{split}
\end{equation*}
where $i\in[n]$, $\ell\in[\tau_i]$ and $U_j \in S_{i}^{(\ell)}$ for some $j \in [n]\backslash\{i\}$.
\label{equation}
\end{problem}

Optimum values for the both objective functions of bi-objective optimization Problem \ref{equation} are plotted as tradeoff curve between $C_s(\vec{\alpha})$ and $C_r(\vec{\beta})$. In this paper the optimization Problem \ref{equation} is solved by weighted sum method for some numeric example.

Some specific cases for optimization Problem \ref{equation} are analyzed in the following subsection.
\subsection{Some Specific Cases}
Considered heterogeneous DSS can be reduced to following cases under some specific restrictions. The cases are as follow:  

	1) (Uniform Reconstruction):
At time $t$, if an arbitrary data collector can retrieve the file by downloading data from exactly $k$ nodes for any combination out of $n$ nodes then the constraint Inequality (\ref{reduced min cut bound}) for the optimization Problem \ref{equation} has additional property $k_t=k, \forall t$.

2) (Uniform Repair Degree):
For a heterogeneous DSS let a node failure can repair by \textit{any} $d$ nodes out of remaining $n-1$ nodes. Under the particular assumption the constraint Inequality (\ref{reduced min cut bound}) for the optimization Problem \ref{equation} reduced to
\begin{problem}
\begin{equation*}
\begin{split}
& \mbox{Minimize: }[C_s(\vec{\alpha}), C_r(\vec{\beta})] \\
& \mbox{subject to} \\
& B\leq\min_{\stackrel{\mathcal{A}_t\in\mathcal{A}}{|\mathcal{A}_t|=k_t}}\sum_{\stackrel{i=1;}{U_{\lambda_i}\in\mathcal{A}_t}}^{k_t}\min\left\{\alpha_{\lambda_i},\sum_{\mu_j}{\beta\left(U_{\lambda_i}, U_{\mu_j}, S_{\lambda_i}^{(\ell)}\right)}\right\}; \\
& 0\leq\alpha_1\leq\alpha_2\leq\ldots\leq\alpha_n; \\
& 1\leq \lambda_1\leq \lambda_2\leq\ldots\leq \lambda_{k_t}\leq n; \\
\end{split}
\end{equation*}
where 
index $\mu_j$ is the index of node $U_{\mu_j}\in S_{\lambda_i}^{(\ell)}\backslash\{U_{\lambda_1},\ldots,U_{\lambda_{i-1}}\}$ such that $\{U_{\lambda_1},\ldots,U_{\lambda_{i-1}}\}\subset S_{\lambda_i}^{(\ell)}$, $\forall i\in[k_t]$ and some $j\in[d]$.
\end{problem}

In this case $|S_m^{(\ell)}|=d$, $\tau_m={n-1 \choose d},\forall m\in[n]$. Here $\min cut$-$capacity(s,D)$ will be given by the node sequence $\left\langle U_{\lambda_i}\right\rangle_{i=1}^{k_t}\in\mathcal{A}(\mathcal{A}_t)$ associated with surviving sequence $\left\langle S_{\lambda_i}^{(\ell)}\right\rangle_{i=1}^{k_t}$ such that $\alpha_{\lambda_1}\leq\alpha_{\lambda_2}\leq\ldots\leq\alpha_{\lambda_{k_t}}$ and $\{U_{\lambda_1},U_{\lambda_2},\ldots,U_{\lambda_{i-1}}\}\subset S_{\lambda_i}^{(\ell)}$.

3) (Uniform Repair Download Amount): 
In this case we assume that downloaded amount from any arbitrary helper node to repair the system is constant say $\beta$. Hence optimization Problem \ref{equation} under the restriction has additional properties as $\beta\left(U_i, U_j, S_{i}^{(\ell)}\right) = \beta$, $\beta\geq 0$ $(\forall i\in[n]$, all possible $j\in[n]\backslash \{i\}$ and $\forall\ell\in[\tau_i])$.



4) (Homogenous DSS):
In heterogeneous DSS become a homogeneous DSS if characteristics of parameters are uniform. Hence assume effective reconstruction degree for any data collector is $k$ and storage capacity of each node is $\alpha$. In addition let a node failure can repair by \textit{any} $d$ nodes out of remaining $n-1$ nodes by downloading $\beta$ packets from each helper node. Under these restrictions, the constraint Inequalities (\ref{reduced min cut bound}) for the optimization Problem \ref{equation} reduced to  
\begin{problem}
\begin{equation*}
\begin{split}
 \mbox{Minimize: }&[C_s(\vec{\alpha}), C_r(\vec{\beta})]\\
 \mbox{subject to}&\\
 B&\leq\sum_{i=1}^{k}\min\left\{\alpha,\left(d-i-1\right)\beta\right\};\\ 
\alpha&\geq 0;\\
\beta&\geq 0.
\end{split}
\end{equation*}
\end{problem}

5) (Other):
In this paper, the considered heterogeneous DSS model can be reduced into some more specific DSS by applying some appropriate restrictions on constraints. 
For example, heterogeneous DSS with uniform reconstruction and uniform repair degree (case $1$ and $2$ respectively) collectively reduces to heterogeneous DSS as investigated in \cite{ETT:ETT2887}.

One can easily find solution of the bi-objective optimization Problem (\ref{equation}) for some numerical values and plot the solution as tradeoff curve for the same. One can compare the tradeoff curve with the tradeoff curve for the existing heterogeneous DSS investigated in \cite{ETT:ETT2887}. Hence in the next section we are calculating some optimum solutions for numerical parameter for our model and comparing it with homogeneous model \cite{capacity} and heterogeneous model \cite{ETT:ETT2887}. 
\subsection{Numerical Work}
For the optimization Problem (\ref{equation}), LP problems with single objective function is solved. The single objective function is calculated by taking linear combination of the two objective functions of optimization problem (\ref{equation}). 
Ten such LP problems are solved by taking distinct linear combination factor between $10^{-3}$ and $10^3$. Plotting tradeoff and solving LP problems are done with the help of `MATLAB' and `lp$\_$solve' \cite{opensource}.

In Figure \ref{Graph 2}, four tradeoff curves are plotted between system repair cost $C_r$ and system storage cost $C_s$ for the respective DSSs. In particular Figure \ref{Graph 2}, one curve is plotted for homogeneous DSS as investigated in \cite{capacity}, another one is drown for a heterogeneous DSS as investigated in \cite{ETT:ETT2887} and remaining two curves are plotted for two heterogeneous DSSs as studied in this paper. In particular, one of the remaining two curves has minimum effective reconstruction degree $k_{\min}$ is $2$ and other has maximum effective reconstruction degree $k_{\max}$ is $2$. For all considered DSSs the common parameters are as follow: $n = 4$, $B = 1$ unit, $\vec{s} = (1\ 10\ 10\ 100)$ and  $\vec{r} = (10\ 1\ 1\ 1)$. For homogeneous DSS and heterogeneous DSS studied in \cite{ETT:ETT2887} have reconstruction degree $k$ = $2$ and repair degree $d$ = $3$. Remaining both heterogeneous DSS have surviving sets $S_1^{(1)}$ = $\{U_2, U_3, U_4\}$, $S_2^{(1)}$ = $\{U_1, U_4\}$, $S_3^{(1)}$ = $\{U_1, U_2\}$, $S_4^{(1)}$ = $\{U_2, U_3\}$.

In Figure \ref{Graph 2}, one can see that our heterogeneous DSS model has more optimum system storage and repair cost then the homogeneous DSS studied in \cite{capacity}. Although the characteristics of our heterogeneous model and heterogeneous model investigated in \cite{ETT:ETT2887} are different, but we obtained some more optimum points for our model as in Figure \ref{Graph 2}. It is shown in the last subsection that one can find heterogeneous DSS considered in \cite{ETT:ETT2887} by taking some restrictions on our model.
\begin{remark}
In the particular tradeoff curves, non-integer solution of bi-objective optimization problem \ref{equation} is also considered. Since the scaling of an arbitrary file size $B$ to $1$, leads to respective integer solution that is not necessarily scale to some integer.
\end{remark}

\begin{figure}
\centering
\includegraphics[scale=0.65]{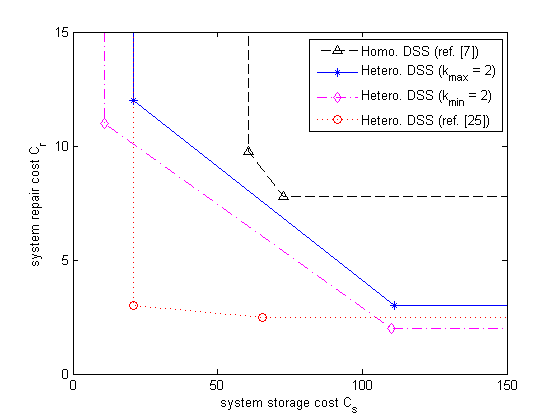}
\caption{
For various DSSs the optimal tradeoff curve is plotted between system repair cost $C_r$ and system storage cost $C_s$. 
}
\label{Graph 2}
\end{figure}
\section{Conclusion}
\label{4}

In this paper, we proposed a model of heterogeneous DSS with dynamic reconstruction degree, storage node capacity and repair bandwidth. In particular, at time $t$, a file can be reconstructed using 
certain set of nodes and  system is repaired for any failed node by contacting some set of helper nodes. For such heterogeneous DSS, the fundamental tradeoff curve between system repair cost and system storage cost is investigated. To plot the tradeoff curve, a bi-objective optimization problem is formulated with the constraints of \textit{min-cut} bound and non-negative parameters of the heterogeneous DSS. The bi-objective optimization problem is solved by weighted sum method for some numerical values of parameters of the heterogeneous model.
Analyzing the tradeoff curve, we observed some more optimum points then the existing heterogeneous model \cite{ETT:ETT2887}. The considered model is close to real world scenario. Our heterogeneous model is  flexible enough to mold it into any existing heterogeneous or homogeneous DSS by considering appropriate restrictions. It would be interesting to construct codes achieving the optimum points on the tradeoff curve.


\bibliographystyle{IEEEtran}
\bibliography{cloud}
\end{document}